\documentclass{cimento}
\usepackage{amssymb}
\usepackage{amsthm}
\usepackage{amsfonts}
\usepackage{graphicx}
\usepackage{pstricks,pst-node,pst-text,pst-3d}
\parskip=\medskipamount
\newtheorem{theorem}{Theorem}

\newtheorem{proposition}[theorem]{Proposition}

\newtheorem{example}{Example}
\def\CA{{\cal A}} \def\CB{{\cal B}}  
\def\CE{{\cal E}}  \def\CG{{\cal G}} 
\def\CI{{\cal I}}   \def\CL{{\cal L}}
 \def\CN{{\cal N}}  
  \def\CS{{\cal S}}

\def\IC{\relax{\rm l\kern-.50 em C}}
\def\IE{\relax{\rm l\kern-.12 em E}}
\def\IK{\relax{\rm l\kern-.18 em K}}
\def\IL{\relax{\rm I\kern-.18 em L}}
\def\IN{\relax{\rm I\kern-.18 em N}}
\def\IR{\relax{\rm I\kern-.18 em R}}

\def\<#1>{\langle#1\rangle}
\def\d<#1>{\langle\langle#1\rangle\rangle}

\def\vfield{{\mathfrak{X}}}

\title{Reduction of Lie-Jordan algebras: Classical}

\author{F. Falceto\from{zar}, 
L. Ferro\from{mad}\from{nap}, 
A. Ibort\from{mad} 
\atque
 G. Marmo\from{mad}\from{nap} 
}

\instlist{
\inst{zar} Departamento de F\'{\i}sica Te\'orica. Universidad de Zaragoza\\ Plaza San Francisco s/n, 50009 Zaragoza, Spain
\inst{mad} Departamento de Matem\'aticas, Universidad Carlos III de Madrid\\ Avda. de la Universidad 30, 28911 Legan\'es, Madrid, Spain
\inst{nap}Dipartimento di Scienze Fisiche, INFN--Sezione di Napoli, Universit\`a di Napoli ``Federico II''\\Via Cintia Edificio 6, I--80126 Napoli, Italy
}

\PACSes{\PACSit{45.20.Jj}{Lagrangian and Hamiltonian mechanics}
 \PACSit{02.40.Yy}{Geometric mechanics}
\PACSit{03.65.Fd}{Quantum Mechanics. Algebraic methods}
}

\begin{document}

\maketitle

\begin{abstract} 
In this paper we present a unified algebraic framework to discuss
the reduction of classical and quantum systems. The underlying algebraic 
structure is a Lie-Jordan algebra supplemented, in the quantum case,
with a Banach structure. We discuss the reduction by symmetries, 
by constraints as well as the possible, non trivial,  combinations of 
both. We finally introduce a new, general framework to perform  
the reduction of physical systems in an algebraic setup.
\end{abstract}


\section{Introduction}
\indent
This paper is the first part of two that jointly study the reduction procedure
in classical and quantum mechanics. One of the difficulties 
when carrying out this program is to find the appropriate common language for both
classical and quantum physics.

Quantum mechanics has been mainly formulated, since its foundations, in
algebraic language \cite{Neu32}.
The observables are elements of a $C^*$ 
algebra and the states are functionals in this space\cite{Ha96}. 
Alternative, geometrical approaches to Quantum
Mechanics have been formulated recently \cite{Ci84}-\cite{ClementeGallardo:2008wa}. 

On the other hand the language of classical mechanics has been mainly geometrical \cite{Abr67}
and in this framework different reduction procedures have been introduced like Marsden-Weinstein reduction,
symplectic reduction, Poisson reduction,... However, it was soon realised that these procedures have their algebraic 
counterpart \cite{Gr94}, \cite{Ib97}, \cite{Cal04}. 

In this contribution we will focus mainly in classical mechanics, while the second part
is devoted to the study of the quantum case. In both parts we will adopt a common algebraic language in terms of Lie-Jordan
algebras, supplemented in the quantum case with a topological Banach space structure.

In the next section we start by a succinct description of Lie-Jordan algebras and its connection with
Poisson algebras and $C^*$ algebras. In the third section we discuss the reduction
procedures for Poisson manifolds in the presence of symmetries or constraints.
Section 4 contains the main results of the paper. In it we introduce a generalisation 
of the previous reductions and it is illustrated
with an example. Finally section 5 is dedicated to the discussion of a possible
extension of the reduction from the classical framework to the quantum one.

\section{Lie-Jordan algebras}

A Lie-Jordan algebra $(\CL,\circ,[\ ,\ ])$ is the combination of
a real, abelian algebra (Jordan algebra) $(\CL,\circ)$ and a Lie algebra $(\CL,[\ ,\ ])$, i.e. 
$[\ ,\ ]$ is an antisymmetric, bilinear bracket that, for any $a, b, c\in \CL$, fulfils
the Jacobi identity
$$[a,[b,c]]+[b,[c,a]]+[c,[a,b]]=0.$$
In addition we require two compatibility conditions between
the two operations: the Leibniz rule
 $$[a\circ b, c]=a\circ[b,c]+[a,c]\circ b,$$
and the associator identity
 $$(a\circ b)\circ c-a\circ(b\circ c)= \hbar^2[[a,c],b],$$
for some  $\hbar\in{\mathbb R}$. Actually if $\hbar\not=0$
we can take it one by an appropriate rescaling of any of the two 
operations  that do not affect the rest of the properties of 
the algebra. The reason why we introduced the constant $\hbar$,
apart from its obvious physical meaning, 
is because in the {\it classical} limit, 
$\hbar=0$, the Jordan algebra becomes associative 
(we shall call it associative Lie-Jordan algebra)
and $(\CA,\circ,[\,,\,])$ is a Poisson algebra.

Notice that given a Lie-Jordan algebra $\CL$, we can define on $\CL^{\mathbb C}$
the following product
$$a\cdot b=a\circ b -i\hbar [a,b]$$
that makes $(\CL^{\mathbb C},\cdot)$ an associative algebra. Moreover
we can introduce the following antilinear involution
$$(a+ib)^*=a-ib,$$
that is an antihomomorphism of the algebra.

Conversely, given a complex associative algebra with 
an antihomomorphism $*$, antilinear and involutive 
$(\CA,\cdot,*)$, the selfadjoint elements  
$${\cal A}_{\rm sa}=\{x\in{\cal A}\vert x^*=x\}$$
with the operations defined by
$$a\circ b=\frac12(a\cdot b+b\cdot a),\qquad
[a,b]=\frac{i}{2\hbar}(a\cdot b-b\cdot a).$$
form a Lie-Jordan algebra
(${\cal A}_{\rm sa},\circ,[\ ,\ ]$)
with $\hbar\not=0$.

In the next section we will discuss at length
the case of associative Lie-Jordan algebras
and its reduction in the context of 
Poisson manifolds.

\section{Reduction of Poisson algebras}

The first example of Poisson algebra that we will consider
is the set of smooth functions in a Poisson manifold $M$,
($C^\infty(M)$, $\circ$, $\{ , \}$), where $\circ$ is the pointwise product
of functions and, for $f,g\in C^\infty(M)$,
$$\{f,g\}=\Pi(df,dg)$$
with $\Pi\in\Gamma(\bigwedge^2 TM)$ such that $[\Pi,\Pi]=0$. Here the square 
brackets represent the Schouten-Nijenhuis bracket for multivector fields.
The latter property guarantees that the Poisson bracket satisfies the Jacobi identity
and all the properties of an associative Lie-Jordan algebra ($\hbar=0$) 
are fulfilled.

The data to construct the Poisson algebra have been give
in terms of geometrical objects, this will be also the 
case when we discuss the reduction procedures. One of the goals 
of the paper is to translate the geometric data
to the algebraic language, in order to compare with the quantum
case, where the discussion is carried out in purely algebraic terms.

\subsection{Reduction by symmetries}

Suppose that we have a Lie group acting on $M$ and we want to restrict our 
Poisson algebra to functions that are invariant under the action of the group.

The infinitesimal action of the group induces a family of vector fields 
$E\subset \vfield(M)$ that we assume to be an integrable distribution.
With these geometric data we introduce the subspace
$$\CE=\{f\in C^\infty(M)\ {\rm s. t.}\ Xf=0,\ \forall X\in\Gamma(E)\}$$
that is a Jordan subalgebra ($\CE\circ\CE\subset\CE$), 
but not necessarily a Lie subalgebra. When this is the case, i.e. if
$$\{\CE,\CE\}\subset \CE,$$
the restrictions of the operations to $\CE$ endows it with 
the structure of a Poisson subalgebra.

From the algebraic point of view the action of vector fields on
functions is a derivation of the Jordan algebra:
$$X(f\circ g)=Xf\circ g+ f\circ Xg,$$ 
and if this derivation is also a Lie derivation:
$$X\{f,g\}=\{Xf,g\}+ \{f, Xg\},$$ 
then one easily sees that $\CE$ is a Lie subalgebra. 

An example of the previous situation is when $E$ is a family of
Hamiltonian vector fields, i.e. there exists a Lie subalgebra
$\cal G\subset C^\infty(M)$ such that $X\in E$ if and only if there is a 
$g\in \CG$ with $Xf=\{g,f\}$ for any $f\in C^\infty(M)$.
This kind of derivations, defined through the Lie product,  
are called inner derivations, they are always Lie derivations and
therefore they define a Lie-Jordan subalgebra 
with the procedure described above. 

\subsection{Reduction by constraints}

In this case the geometric input is a submanifold
$N\subset M$ and the goal is to define a Poisson algebra in
the set of smooth functions on $N$ or, at least, in a 
subset of it.

In order to carry out the algebraic reduction we introduce
the Jordan ideal of functions that vanish on $N$,
$${\cal I}=\{f\in C^\infty(M)\ {\rm s. t.}\  f\vert_N=0\}.$$
Then, the Lie normaliser of $\CI$,
$${\cal N}=\{g\in C^\infty(M)\ {\rm s. t.}\ \{g, {\cal I}\}\subset {\cal I}\},$$
is a Lie-Jordan subalgebra, as an straightforward computation shows, and 
$\CN\cap\CI$ is its Lie-Jordan ideal.
Therefore ${\cal N}/({\cal N}\cap{\cal I})$ inherits the structure of a Lie-Jordan 
algebra. 

The reduced algebra as written above does not seems to have a direct connection with 
the functions on $N$. In order to uncover this connection we use the second isomorphism
theorem for vector spaces
$${\cal N}/({\cal N}\cap{\cal I})\simeq({\cal N}+{\cal I})/{\cal I}$$
and taking into account that the quotient by $\CI$ can be identified with the restriction to 
$N$ the right hand side can be described as the restriction to $N$ of the functions in 
$\CN+\CI$. If $\CN+\CI=C^\infty(M)$ (the constraints are second class in Dirac's terminology \cite{Dirac}) we obtain a 
Poisson algebra structure in $C^\infty(N)$. The Poisson bracket, in this case, is restriction to $N$ of the Dirac 
bracket \cite{Dirac} in $M$ determined by the second class constraints.

\section{More general Poisson reductions}

One attempt to combine the previous two reductions to define a more general one
is  contained in \cite{Ma86}. We shall rephrase in algebraic terms the
original construction that was presented in geometric language.

The data are an embedded submanifold $\iota:N\rightarrow M$ of a Poisson manifold
and a sub\-bun\-dle $B\subset T_NM := \iota^*(TM)$. With these data we define
the Jordan ideal ${\cal I}=\{f\in C^\infty(M)\ {\rm s. t.}\  f\vert_N=0\},$
as before and the Jordan subalgebra 
${\cal B}=\{f\in C^\infty(M)\ {\rm s. t.}\ Xf=0\ \forall X\in \Gamma(B)\}$.
The goal is to define an associative Lie-Jordan structure in 
$\CB/(\CB\cap\CI)$.

Following \cite{Ma86} we assume that $\CB$ is also a Lie subalgebra, then
if $\CB\cap\CI$ is a Lie ideal of $\CB$ the sought reduction is possible.

However, the condition that $\CB$ is a subalgebra is a rather 
strong one\cite{Fa08} and, consequently,
the reduction procedure is much less general than initially expected. 
Actually, as we will show, it consists on a succesive application of the
reductions introduced in the previous section. 
One can prove the following result.

\begin{theorem}
With the previous definitions, if $\CB$
is not the whole algebra, i.e. $B\not= 0$, 
and in addition it is a Lie subalgebra, then the following hold:
\begin{itemize}
\item[a)] {$\CB\subset\CN:=\{g\in C^\infty(M)\ {\rm s. t.} \{\CI,g\}\subset\CI\}$.}
\item[b)] $\CB\cap\CI$ is Poisson ideal of $\CB$.
\item[c)] $\CB/(\CB\cap \CI)$ always inherits a Poisson bracket.
\item[d)] Take another $0\not=B'\subset T_N(M)$ and define $\CB'$ accordingly.
If $B\cap TN=B'\cap TN$ $\Leftrightarrow$ $\CB+\CI=\CB'+\CI$ by the second 
isomorphism theorem we have
 $$\CB/(\CB\cap\CI)\simeq{(\CB+\CI)/\CI}\simeq\CB'/(\CB'\cap\CI)$$
and the two Poisson brackets induced on $(\CB+\CI)/\CI$ coincide.
\end{itemize} 
\end{theorem}

\begin{proof} We prove a) by contradiction. Assume that $\CB\not\subset\CN$ then
there exist functions $f\in\CB$, $g\in\CI$ and an open set $U\subset N$, such that
$$\{g,f\}(p)\not=0,\quad {\rm for\ any}\ p\in U.$$
But certainly $g^2\in\CB$ as a simple consequence of the Leibniz rule for the action
of vector fields. Therefore, using that $\CB$ is a Lie subalgebra we have  
$$ \{g^2,f\}=2g\{g,f\}\in\CB$$
and due to the fact that $g\in \CI$ and $\{g,f\}(p)\not=0$ this implies
$g\in\CB_U$, where $\CB_U$ is the set of functions whose restriction to $U$ coincide with
the restriction of someone in $\CB$.

So far we know that $g\in\CB_U\cap\CI$ and therefore $hg\in\CB_U\cap\CI$ for any $h\in C^\infty(M)$.
But using that $\CB_U$ is a Lie subalgebra as it is $\CB$ (due to the local character of the Poisson bracket) 
we have
$$ \{hg,f\}=h\{g,f\}+g\{h,f\}\in\CB_U\Rightarrow h\{g,f\}\in\CB_U\Rightarrow  h\in\CB_U.$$
But $h$ is any function, then $\CB_U=C^\infty(M)$ and $B\vert_U=0$ which implies
$B=0$ as we assumed that it is a subbundle. This contradicts the hypothesis of the theorem
and a) is proved.

b) follows immediately from a). Actually if $\CB\subset\CN$ we have $\{\CI,\CB\}\subset\CI$
and moreover $\{\CB,\CB\}\subset\CB$. Then $\{\CI\cap\CB,\CB\}\subset\CI\cap\CB$.

c) is a simple consequence of the fact that $\CB$ is a Lie-Jordan subalgebra and $\CB\cap\CI$ its 
Lie-Jordan ideal.

To prove d) take $f_i\in\CB$ and $f'_i\in\CB'$, $i=1,2$, such that $f_i+\CI=f'_i+\CI$.
The Poisson bracket in $(\CB+\CI)/\CI$ is given by
$$\{f_1+\CI,f_2+\CI\}=\{f_1,f_2\}+\CI\in (\CB+\CI)/\CI,$$
where for simplicity we use the same notation for the Poisson bracket in the
different spaces, which should not lead to confusion.
We compute now the alternative expression $\{f'_i+\CI,f'_2+\CI\}=\{f'_1,+f'_2\}+\CI$.
We assumed $f'_i=f_i+g_i$ with $g_i\in\CI\cap(\CB+\CB')$ and 
therefore, as a consequence of a), we have  $\{f_1,g_2\},\{g_1,f_2\},\{g_1,g_2\}\in\CI$, 
which implies  
$$\{f'_1,+f'_2\}+\CI=\{f_1,f_2\}+\CI$$
and the proof is completed.
\end{proof}

Last property implies that the reduction process does not depend effectively on $B$ 
but only on $B\cap TN$. 
Actually one can show that this procedure is simply a successive application of the two previous reductions
presented before: first we reduce the Poisson bracket by constraints
to $N$ and then by symmetries with $E=B\cap TN$.

For completeness we would like to comment on the situation when $B=0$.
In this case $\CB=C^\infty(M)$ and, of course, it is always a Lie subalgebra. 
Under these premises the reduction is not possible unless $\CI$ is a Lie ideal
which is not the case in general. Anyhow, if the conditions to perform the reduction
are met and we consider some  $B'\not=0$
such that $B'\cap TN=0$ and $\CB'$ is a Lie subalgebra,
then we obtain again property d) of the theorem: the Poisson brackets 
induced by $B=0$ and $B'$ on  $\CB/\CI$ are the same.  

The question then is if given $N$ and $B$ there is a more general way to obtain the desired 
associative Lie-Jordan structure in $\CB/(\CB\cap\CI)$ where $\CB$ and $\CI$ are defined as 
before. 

To answer this question we will rephrase the problem in purely algebraic terms.  We shall assume that together 
with an associative Lie-Jordan algebra we are given a Jordan ideal $\CI$ and a Jordan subalgebra $\CB$. 
Of course, a particular example of this is the geometric scenario discussed before. 
Under these premises $\CB\cap\CI$ is a Jordan ideal of $\CB$ and $\CB+\CI$ is a Jordan subalgebra,
then it is immediate to define Jordan structures 
on $\CB/(\CB\cap\CI)$ and on $(\CB+\CI)/\CI$ such that the corresponding projections $\pi_B$ and $\pi$ are Jordan homomorphisms. 
Moreover,  the natural isomorphism between both spaces is also a Jordan 
isomorphism. The problem is 
whether we can also induce a Poisson bracket in the quotient spaces. 
One first step to carry out this program is contained in the following theorem.

\begin{theorem}\label{th:genred}
Given an associative Lie-Jordan algebra, $(\CL,\circ,\{\ ,\ \})$, 
a Jordan ideal $\CI$ and a Jordan subalgebra $\CB$, assume
\begin{equation}\label{weakcond}
{\rm a)}\ \{\CB,\CB\}\subset \CB+\CI,\qquad {\rm b)}\ \{\CB,\CB\cap\CI\}\subset\CI,
\end{equation}
then the following commutative diagram
\begin{equation}\label{inducedpoisson}
\begin{psmatrix}[mnode=R,colsep=1.5cm,rowsep=1.5cm]
\CB\times\CB && \CB+\CI\\
\CB/(\CB\cap\CI)\times\CB/(\CB\cap\CI)&
\CB/(\CB\cap\CI)&(\CB+\CI)/\CI
\end{psmatrix}
\psset{nodesep=0.3cm}
\everypsbox{\scriptstyle}
\ncLine{->}{1,1}{1,3}\Aput{\{\ ,\ \}}
\ncLine{->}{1,3}{2,3}\Bput{\pi}
\ncLine{->}{1,1}{2,1}\Aput{\pi_B\times\pi_B}
\ncLine{<->}{2,2}{2,3}\Aput{\simeq}
\ncLine{->}{2,1}{2,2}
\end{equation}
defines a unique bilinear, antisymmetric operation in $\CB/(\CB\cap\CI)$
that satisfies the Leibniz rule. 
\end{theorem}

\begin{proof} In order to show that we define uniquely an operation we have to check that
$\pi_B$ is onto and that $\ker(\pi_B)\times\CB$ and $\CB\times \ker(\pi_B)$ are mapped into $Ker(\pi)=\CI$. 
But first property holds because $\pi_B$ is a projection and the second one is a consequence of 
(\ref{weakcond},b). The bilinearity of the induced operation follows form
the linearity or bilinearity of all the maps involved in the diagram
and its antisymmetry derives form that of $\{\ ,\ \}$.
Finally Leibniz rule is a consequence of the same property for the original Poisson bracket
and the fact that $\pi$ and $\pi_B$ are Jordan homomorphisms.
\end{proof}

The problem with this construction is that, in general, the bilinear operation does not
satisfy the Jacobi identity as shown in the following example.

\begin{example}\label{ejemplouno}
Consider $M={\mathbb R}^3\times{\mathbb R}^3$, with coordinates $({\bf x}, {\bf y})$
and Poisson bracket given by the bivector
$\Pi=\sum_{i=1}^3\frac{\partial}{\partial_{x_i}}\wedge\frac{\partial}{\partial_{y_i}}$.
take $N=\{(0,0,x_3,{\bf y})\}$ and for a given $\lambda\in C^\infty(N)$
define $B={\rm span}\{\partial_{x_1},\partial_{x_2}-\lambda\partial_{y_1}\}
\subset T_N M$ and 
$$\CB=\{f\in C^\infty(M),\ {\rm s. t.}\ X f|_N=0, \forall\ X\in\Gamma (B)\}.$$
Notice that $T_NM$ is a direct sum of $B$ and $TN$,
therefore
we immediately get
$$\{\CB,\CB\}\subset\CB+\CI=C^\infty(M)\quad{\rm and}\quad\{\CB,\CB\cap\CI\}\subset\CI,$$ 
and we meet all the requirements to define a bilinear, antisymmetric operation on
$\CB/(\CB\cap\CI)\simeq C^\infty(N)$.

Using coordinates $(x^3,{\bf y})$ for $N$ the bivector field is
$$\Pi_N=\frac{\partial}{\partial_{x_3}}\wedge\frac{\partial}{\partial_{y_3}}+
\lambda \frac{\partial}{\partial_{y_1}}\wedge\frac{\partial}{\partial_{y_2}}$$
that does not satisfy the Jacobi identity unless 
$\partial_{x_3}\lambda=\partial_{y_3}\lambda=0$.
\end{example}

Now the problem is to supplement (\ref{weakcond}) with more conditions
to guarantee that the induced operation satisfies all the requirements
for a Poisson bracket. We do not know a simple description of the minimal 
necessary assumption but a rather general scenario is the following proposition:

\begin{proposition}
Suppose that in addition to the conditions of theorem \ref{th:genred} 
we have two Jordan subalgebras $\CB_+$, $\CB_-$ 
$${\CB_-}\subset\CB\subset{\CB_+}\quad{\it and}\quad  
{\CB_\pm+\CI}=\CB+\CI,$$
such that 
\begin{equation}\label{strongcond}
{\rm a)}\ 
\{\CB_-,\CB_-\}\subset\CB_+,\quad{\rm b)}\ \{\CB_-,\CB_+\cap\CI\}\subset \CI.
\end{equation}
Then the antisymmetric, bilinear operation induced by (\ref{inducedpoisson})  
is a Poisson bracket, i.e. it fulfils the Jacobi identity. 
\end{proposition}

\begin{proof} To prove this statement consider any two functions 
$f_1,f_2\in\CB$ and, for $i=1,2$, denote by $f_{i,-}$ a function
in $\CB_-$ such that $f_i+\CI=f_{i,-}+\CI\subset \CB+\CI$.
Due to (\ref{weakcond}) we know that 
$$\{f_{1,-},f_{2,-}\}+\CI= \{f_{1},f_{2}\}+\CI,$$
but if (\ref{strongcond},a) also holds,  
$$\{f_{1,-},f_{2,-}\}\in \CB_+,$$
in addition we have that
$$\{f_{1,-},f_{2,-}\}_--\{f_{1,-},f_{2,-}\}\in \CB_+\cap\CI,$$
and using (\ref{strongcond},b)
$$\{\{f_{1,-},f_{2,-}\}_-,f_{3,-}\}+\CI
=
\{\{f_{1,-},f_{2,-}\},f_{3,-}\}+\CI.$$
Therefore the Jacobi identity for the reduced antisymmetric product
derives from that of the original Poisson bracket.
\end{proof}

Notice that the whole construction has been made in algebraic terms
and therefore it will have an immediate translation to the quantum 
realm. But before going to that scenario we reexamine
the example to show how it fits into
the general result.

\begin{example}
We take definitions and notations from example \ref{ejemplouno}.
Now let $\tilde\lambda$ be an arbitrary smooth extension of $\lambda$ to $M$, i.e.
$\tilde\lambda\in C^\infty(M)$ such that $\tilde\lambda\vert_N=\lambda$,
we define $E={\rm span}\{\partial_{x_1},\partial_{x_2}-\tilde\lambda\partial_{y_1}\}
\subset T M$ and $\CB_-=\{f\in C^\infty(M)\ {\rm s. t.}\ Xf=0,\ \forall X\in\Gamma(E)\}$.

If we define $\CB_+=\CB$, it is clear that $\CB_-\subset\CB\subset\CB_+$,
$\CB_\pm+\CI=\CB+\CI$ and $\{\CB_-,\CB_+\cap\CI\}\subset\CI$. 
But $\{\CB_-,\CB_-\}\subset\CB_+$ if and only if $\partial_{x_3}\lambda=\partial_{y_3}\lambda=0$.
\end{example}

Therefore, in our construction we can accommodate the most general situation 
in which the example
provides a Poisson bracket. 
We believe that this is not always the case, but we do not have any 
counterexamples.

\section{Final comments}
We want to end this contribution with a comment on the possible application of the 
reduction described in the previous section to quantum systems. 
In this case the Lie-Jordan algebra is non-associative and due to the associator identity
there is a deeper connection between the Jordan and Lie products. 
As a result the different treatment between the Jordan and the Lie
part, that we considered in the case of associative algebras, is not useful any more
and the natural thing to do is to consider a more {\it symmetric} prescription.

We propose a generalisation of the standard reduction procedure 
(the quotient of subalgebras by ideals) along similar lines to those
followed in the associative case. 

The statement of the problem is the following: given a Lie-Jordan algebra $\CL$ and two 
subspaces $\CB$, $\CS$ the goal is to induce a Lie-Jordan structure in the quotient
$\CB/(\CB\cap\CS)$.  

If we assume
\begin{eqnarray*}
\CB\circ\CB\subset\CB+\CS,\quad &\CB\circ(\CB\cap\CS)\subset\CS,\cr
[\CB,\CB]\subset\CB+\CS,\quad & [\CB,\CB\cap\CS]\subset\CS,
\end{eqnarray*}
then a diagram similar to (\ref{inducedpoisson}) allows to induce
commutative and anticommutative, bilinear operations in the quotient.
Now, in order to fulfil the ternary properties 
(Jacobi, Leibniz and associator identity) we need more conditions.
We can show that, again, it is enough to have two more subspaces 
$\CB_-\subset\CB\subset\CB_+$ such that $\CB_\pm+\CS=\CB+\CS$ and moreover
\begin{eqnarray*}
\CB_-\circ\CB_-\subset\CB_+,\quad&
\CB_-\circ( \CB_+\cap\CS)\subset\CS,\cr
[\CB_-,\CB_-]\subset\CB_+,\quad&
[\CB_-,(\CB_+\cap\CS)]\subset\CS.
\end{eqnarray*}
Then, under these conditions, one can correctly induce a Lie-Jordan structure 
in the quotient.

There are at least two aspects of this construction that need more work.
The first one is to find examples
in which this reduction procedure is relevant, similarly to what we did 
for the classical case in the previous section. 
The second problem is of topological nature:
given a Banach space structure in the big algebra $\CL$, compatible with its operations,
we can correctly induce a norm in the quotient provided $\CB$ and $\CS$ are closed subspaces.
However, the induced operations need not to be continuous in general;
though they are, if $\CB$ is a subalgebra and $\CS$ an ideal \cite{FFIM13}. 
The study of more general conditions for continuity
and compatibility of the norm will be the subject 
of further research.
    
\acknowledgments
 This work was partially supported by MEC grants FPA--2009-09638,
MTM2010-21186-C02-02, QUITEMAD programme and DGA-E24/2.  
G.M. would like to acknowledge the support provided by the Santander/UCIIIM Chair of Excellence programme 2011-2012. 
\vskip 1cm

\end{document}